\documentclass{article}

\usepackage{microtype}
\usepackage{graphicx}
\usepackage{subfigure}
\usepackage{booktabs} 

\usepackage{hyperref}

\usepackage{tabularx}
\usepackage{booktabs} 
\usepackage{multirow}
\usepackage{tcolorbox}
\usepackage{float}
\usepackage{graphicx}

\usepackage[accepted]{icml2025}

\usepackage{amsmath}
\usepackage{amssymb}
\usepackage{mathtools}
\usepackage{amsthm}
\usepackage{algorithm}

\usepackage[capitalize,noabbrev]{cleveref}

\theoremstyle{plain}
\newtheorem{theorem}{Theorem}[section]

\newtheorem{corollary}[theorem]{Corollary}
\theoremstyle{definition}

\newtheorem{assumption}[theorem]{Assumption}
\theoremstyle{remark}

\usepackage[textsize=tiny]{todonotes}

\icmltitlerunning{CARTS: Collaborative Recommendation Agent Framework for Titles}

\begin{document}

\twocolumn[
\icmltitle{CARTS: Collaborative Agents for Recommendation Textual Summarization
}



\icmlsetsymbol{equal}{*}

\begin{icmlauthorlist}
\icmlauthor{Jiao Chen}{equal,Bellevue}
\icmlauthor{Kehui Yao}{equal,Bellevue}
\icmlauthor{Reza Yousefi Maragheh}{equal,Sunnyvale}
\icmlauthor{Kai Zhao}{equal,Sunnyvale}
\icmlauthor{Jianpeng Xu}{Sunnyvale}
\icmlauthor{Jason Cho}{Sunnyvale}
\icmlauthor{Evren Korpeoglu}{Sunnyvale}
\icmlauthor{Sushant Kumar}{Sunnyvale}
\icmlauthor{Kannan Achan}{Sunnyvale}
\end{icmlauthorlist}

\icmlaffiliation{Bellevue}{Walmart Global Tech, Bellevue, WA, US}
\icmlaffiliation{Sunnyvale}{Walmart Global Tech, Sunnyvale, US}

\icmlcorrespondingauthor{Jiao Chen}{jiao.chen0@walmart.com}
\icmlcorrespondingauthor{Kehui Yao}{kehui.yao@walmart.com}
\icmlcorrespondingauthor{Reza Yousefi Maragheh}{reza.yousefimaragheh@walmart.com}
\icmlcorrespondingauthor{Kai Zhao}{kai.zhao@walmart.com}

\icmlkeywords{Machine Learning, ICML}

\vskip 0.3in
]


\printAffiliationsAndNotice{\icmlEqualContribution} 

\begin{abstract}
Current recommendation systems often require some form of textual data summarization, such as generating concise and coherent titles for product carousels or other grouped item displays. While large language models have shown promise in NLP domains for textual summarization, these approaches do not directly apply to recommendation systems, where explanations must be highly relevant to the core features of item sets, 
adhere to strict word limit constraints. In this paper, we propose CARTS (Collaborative Agents for Recommendation Textual Summarization), a multi-agent LLM framework designed for structured summarization in recommendation systems. CARTS decomposes the task into three stages—Generation Augmented Generation (GAG), refinement circle, and arbitration, where successive agent roles are responsible for extracting salient item features, iteratively refining candidate titles based on relevance and length feedback, and selecting the final title through a collaborative arbitration process. Experiments on large-scale e-commerce data and live A/B testing show that CARTS significantly outperforms single-pass and chain-of-thought LLM baselines, delivering higher title relevance and improved user engagement metrics. 
\end{abstract}

\section{Introduction}

\begin{figure*}
    \centering
    \includegraphics[width=1.0\linewidth]{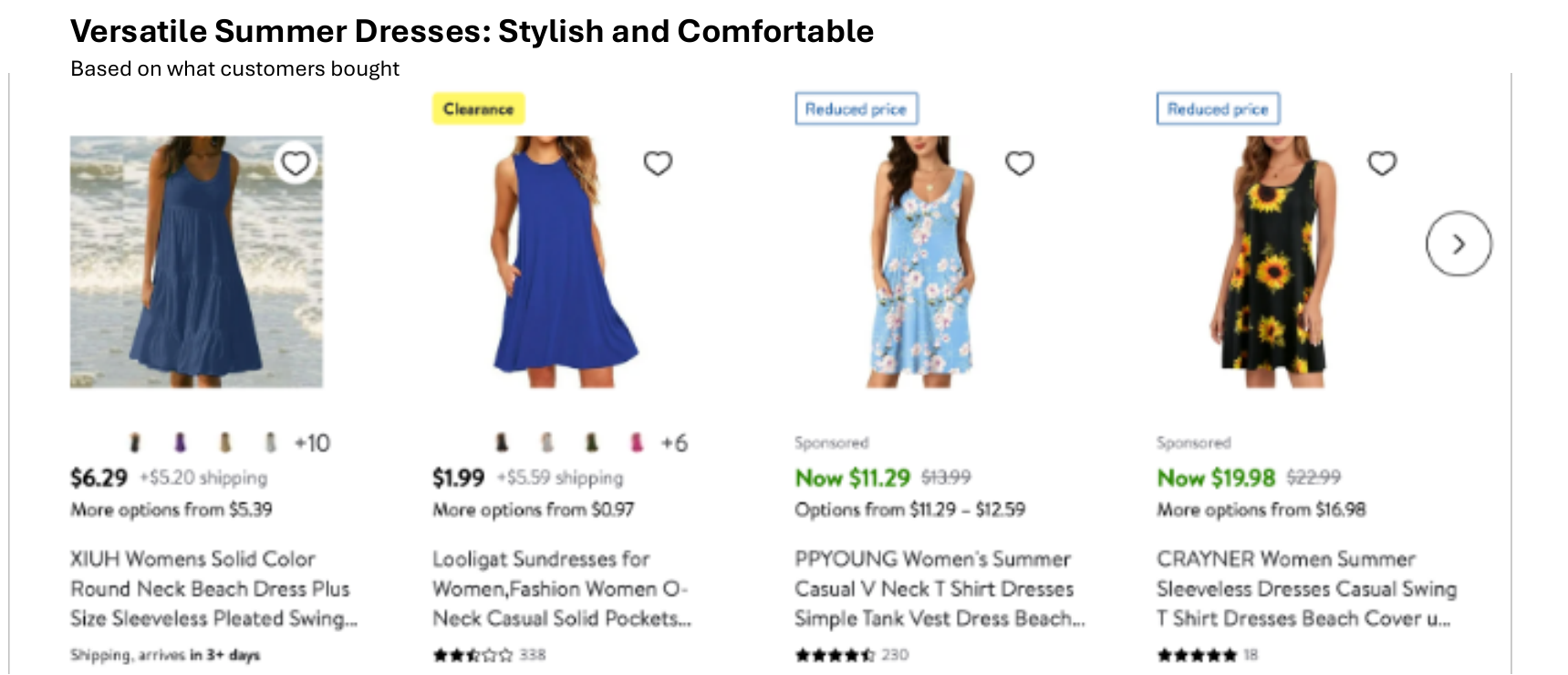}
    \caption{Example of a Product Carousel with Human-Curated Module Title}
    \label{fig:summer}
\end{figure*}

Modern recommendation systems increasingly rely on textual summarization to enhance transparency, usability, and engagement. A common example is the need to generate concise and coherent titles for grouped item displays—such as product carousels or recommendation modules—that communicate the shared theme or purpose of the items shown \cite{ge2022survey}. These summary titles serve not only to improve user understanding but also to drive attention and interaction within limited user interface space \cite{zhang2020explainable}. 
Figure \ref{fig:summer} illustrates such a scenario in e-commerce, where a module displays several sundresses. The accompanying module title—“Versatile Summer Dresses: Stylish and Comfortable”—serves as a human-readable summary of the visual and semantic commonalities across the items.


Automatically generating such titles presents unique challenges in semantic aggregation, coverage, and conciseness. Multi-agent LLM frameworks have shown promise in domains like law, finance, and long-form summarization—where agents collaboratively process lengthy, unstructured documents—their applicability to recommendation systems remains limited and underexplored. In contrast to these domains, recommendation explanation involves summarizing structured item metadata (e.g., titles, categories, reviews) across diverse products into a single, short, behaviorally effective title \cite{wang2024survey, liang2023encouraging, du2023improving}. This task introduces unique challenges: explanations must be generated under strict length constraints, reflect semantic overlap across multiple items, and align with business objectives such as user engagement and conversion. Existing agentic frameworks are typically optimized for coherence and factuality, not for maximizing item coverage or optimizing for click-through rates in user interfaces. 


To fill the research gap above, 
we propose CARTS (Collaborative Recommendation Agent Framework for Titles), a multi-agent LLM-based framework for generating module-level recommendation explanations.  This collaborative process enables more accurate, diverse, and UI-aligned summaries compared to single-agent LLM baselines. CARTS decomposes the task into three stages. First, in the \emph{Generation Augmented Generation (GAG)} stage, a two-step process is used to distill key item features and synthesize initial candidate titles. Next, during the \emph{Refinement Circle}, feedback agents iteratively critique each title with respect to item relevance, stylistic constraints, and length limitations. These critiques are then used to guide the generator in producing progressively improved title versions across multiple iterations. Finally, in the \emph{Arbitration} stage, an Arbitrator selects the final title that best balances semantic relevance and constraint satisfaction. 

In addition to its architectural novelty, CARTS provides a theoretical foundation for the agent collaboration process by framing title generation as a constrained coverage optimization problem. We theoretically analyze the \emph{Refinement Circle} and derive approximation guarantees on the number of refinement steps required to reach near-optimal coverage under practical character-length constraints. Specifically, under assumptions of reliable feedback and generator agents, we prove that CARTS can achieve a desired fraction of optimal relevance with high probability in a bounded number of iterations. This theoretical insight not only differentiates CARTS from prior heuristic agent workflows, but also grounds its design in convergence-aware optimization.

We evaluate CARTS through extensive offline experiments and online A/B tests on a real-world e-commerce platform. Results demonstrate that CARTS improves both module-level relevance scores and business outcomes such as click-through rate (CTR), add-to-cart rate (ATCR), and gross merchandise value (GMV), highlighting the promise of multi-agent LLM systems in real-world recommendation workflows.

Our contribution are summarized as follows: 
\begin{itemize}

    \item We propose CARTS, a novel multi-agent LLM framework that integrates generation, refinement and arbitration to collaboratively generate concise module titles that maximize item-level relevance coverage under practical constraints.
    
    \item We provide a theoretical analysis of CARTS’s \emph{Refinement Circle}, proving an approximation guarantee on the number of steps required to achieve near-optimal item coverage under length constraints, based on reliability assumptions of the feedback and generator agents.

    \item We demonstrate the empirical effectiveness of CARTS through comprehensive offline experiments and online A/B tests, showing consistent improvements in both explanation quality and commercial impact.

\end{itemize}

\begin{figure*}[h]
    \centering
    \includegraphics[width=0.95\linewidth]{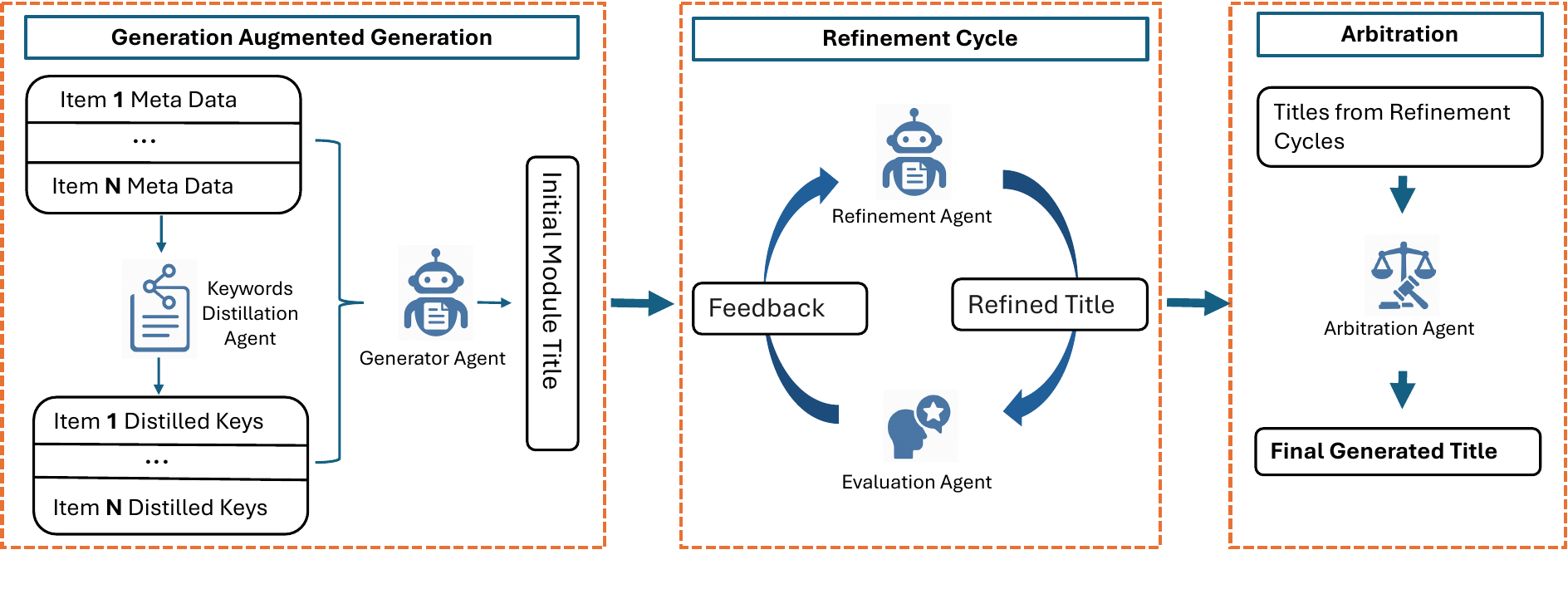}
    \caption{The overview of proposed multi-agent methods for module title generation.}
    \label{fig:overview}
\end{figure*}

\section{Related Works}
Recommendation systems are crucial in various industries, retrieving relevant documents for users based on context. Recommendation explaination can help reduce customer confusion and abandonment \cite{ge2022survey}, making it important to explain recommendations and align them with explanations \cite{zhang2020explainable}.



Classic recommendation explanation approaches often use surrogate models trained alongside the primary recommendation model \cite{catherine2017explainable, chen2021generate, wang2018reinforcement}. Recent advancements in Large Language Models (LLMs) have extended their applicability to various domains \cite{maragheh2023llm, maragheh2023llm2, chen2023knowledge, wei2024llmrec}, including recommendation explanation \cite{lei2023recexplainer, gao2024dre}. However, default LLM-based frameworks may struggle with complex tasks, potentially leading to hallucinations \cite{huang2023survey}. To address this, agentic frameworks have emerged, combining LLMs with planning and memory modules to enhance performance and execute complex tasks more effectively \cite{wang2024survey, zhang2024survey}.

The most simple agentic frameworks use a single agent to complete the entire sequence of tasks. While effective in many cases \cite{wang2023c, zhang2023generative}, single-agent frameworks may struggle to handle highly complex, multi-faceted tasks due to their lack of specialization and inability to multitask effectively. Generating recommendation explanations, for instance, requires satisfying multiple, often competing objectives, such as ensuring high relevance to recommendations while being concise, persuasive, and transparent to customers. Single-agent approaches may not be able to sufficiently balance these diverse requirements. Multi-agent frameworks, by contrast, uses the collective intelligence of individual specialized LLM agents, are capable of imitating complex settings where humans engage in collaboration to achieve a common goals, through planning, discussions, decision making,  task conduction, and evaluation \cite{wang2024multi, liang2023encouraging, du2023improving}.

Different from existing recommendation explanation studies, which usually focus on  explaining the single recommended item, CARTS introduces a novel multi-agent framework specifically tailored for the task of generating unified and compelling carousel titles for a list of recommended items, aiming to improve module transparency and increase customer engagement. Previous existing multi-agent systems mainly focus on general problem-solving or conversational collaboration, CARTS deploys specialized LLM agents that iteratively refine the generated module titles to address the diverse requirements in effective eCommerce carousel titling, leading to a significant improvement in title quality compared to single-agent or simpler sequential approaches.

\section{Methodology}

\label{sec:methodology}

In this section, we formally describe the proposed \textbf{CARTS} framework. Our approach consists of three main components: (i) \emph{Generation Augmented Generation (GAG)} (ii) \emph{Refinement Circle} (iii) \emph{Arbitration}. Figure~\ref{fig:overview} illustrates the entire pipeline.

\subsection{Problem Definition}
\label{sec:problem_def}
Suppose a recommendation system produces a list of $N$ items: $\mathcal{I} \;=\; \{ I_1, I_2, \dots, I_N \} $ to be displayed in a carousel interface. Each item $I_i$ contains textual metadata such as: (i) Catalog Information $C_i$ (e.g., product categories or sub-categories), (ii) Title Text $T_i$, Supplementary Text $P_i$ (e.g., seller descriptions or customer reviews). 
We write $I_i = (C_i, T_i, P_i)$ for $i = 1, 2, \dots, N$. Our goal is to generate a succinct and persuasive \emph{module title}, denoted $M_{\mathrm{title}}$, that highlights the shared use cases, advantages, or attributes among the $N$ recommended items. Formally, we seek a function $\mathrm{GEN}$ such that $M_{\mathrm{title}} \;=\; \mathrm{GEN}\bigl(I_1, I_2, \dots, I_N\bigr)$.

In addition to simply generating a title, we further desire that $M_{\mathrm{title}}$ to be (i) relevant to as many items in $\{I_1, I_2, \ldots, I_N\}$ as possible, and (ii) satisfy a given set of imposed constraints on the written text (for instance to respects a practical length limit, such as a maximum of $K$ characters--reflecting UI constraints or readability considerations).

Let us define a relevance indicator, $R\bigl(M_{\mathrm{title}}, I_i\bigr)$, which evaluates how well the title corresponds to item $I_i$. For each item, $R$ outputs $1$ if $M_{\mathrm{title}}$ is deemed relevant, and $0$ otherwise.\footnote{In practice, $R$ could be a more nuanced function (e.g., a continuous measure of semantic overlap), but here we use a binary indicator for simplicity.} By assuming at least one constraint for character length, we pose the selection of $M_{\mathrm{title}}$ as the following optimization problem:
\begin{align}
\max_{M_{\mathrm{title}}} \quad & \sum_{i=1}^{N} R\bigl(M_{\mathrm{title}}, I_i\bigr) \label{eq:obj:max-items} \\
\text{subject to} \quad & \mathrm{len}\bigl(M_{\mathrm{title}}\bigr) \le K, \label{eq:constr:length} \\
& M_{\mathrm{title}} \in \mathcal{C} \label{eq:constr:others}
\end{align}

where $\mathrm{len}\bigl(M_{\mathrm{title}}\bigr)$ denotes the character length of the proposed title, and $K$ is a threshold specified by the interface constraints. The set $\mathcal{C}$ (constraint~\eqref{eq:constr:others}) can represent any additional conditions, such as stylistic guidelines or language appropriateness rules.

In summary, we wish to choose a title $M_{\mathrm{title}}$ that maximizes the count of items for which the title is relevant, while ensuring the title length remains below $K$. This setup allows for diverse scoring functions or additional constraints (e.g., word-based or phrasing constraints) as required by the application. In the subsequent sections, we discuss how our LLM-based multi-agent framework (CARTS) is equipped to implicitly balance these objectives and constraints when generating, refining, and selecting the final module title.

\subsection{Generation Augmented Generation (GAG)}
\label{sec:gag}
A naive large language model (LLM) approach is to concatenate $I_1, I_2, \dots, I_N$ into a single prompt and directly request a module title. As discussed in Section\,\ref{sec:problem_def}, our goal is not only to generate a concise title but also to \emph{maximize the number of items} for which this title is relevant, all while satisfying a character-length constraint (see Eq.\,\ref{eq:obj:max-items}--\ref{eq:constr:length}). Direct prompting often yields irrelevant or under-specified titles that fail to cover the shared attributes or exceed length limits. Hence, we propose \emph{Generation Augmented Generation} (\textbf{GAG}) in two steps, designed to spotlight each item’s most salient features and thereby improve coverage within allowable length bounds.

\subsubsection{Distillation of Keywords}
In the first step, we focus on extracting a small set of keywords for each item to highlight its essential features. Let $\mathrm{DISTILL}$ be an LLM-driven function that produces $l$ keywords from the catalog information, title text, and supplementary text of item $I_i$:
\begin{equation}
\label{eq:distill}
    \{ K_i^1, K_i^2, \ldots, K_i^l \} 
    \;=\; \mathrm{DISTILL}(C_i, T_i, P_i),
\end{equation}
where $K_i^j$ is the $j$-th keyword for item $I_i$. This keyword set allows the subsequent generation stage to more effectively capture overlapping aspects across all items, thus increasing the potential for broad relevance.

\subsubsection{Title Generation with Augmented Prompts}
Next, we augment the original metadata $I_i$ with the distilled keyword set $\{ K_i^1, \dots, K_i^l \}$. Define $G_{\mathrm{title}}$ as an LLM-based function that produces a concise module title from the augmented prompt:
\begin{equation}
\label{eq:aug_title}
    M_{\mathrm{title}}^{(0)}
    \;=\;
    G_{\mathrm{title}}\Bigl(
    \{(I_1, K_1^{1:l}), \dots, (I_N, K_N^{1:l})\}
    \Bigr).
\end{equation}
Here, $M_{\mathrm{title}}^{(0)}$ is the \emph{initially generated} title, which should ideally cover the shared attributes of $\{I_1, \dots, I_N\}$ as much as possible while remaining within the length limit $K$. In practice, the prompt can include explicit reminders (e.g., ``The title must not exceed $K$ characters and should be relevant to as many items as possible'') to help the LLM internalize these constraints.

\subsection{Refinement Circle}
\label{sec:refine}
To further improve coverage and manage the length constraint, we introduce a refinement loop involving multiple LLM agents:
\begin{enumerate}
    \item A \textbf{generator agent}, which proposes a candidate title based on the item information and the feedback.
    \item A \textbf{feedback agent}, which evaluates the candidate title and provides natural language feedback to the item list. For example: (i) whether the title adheres to the character limit $K$, and (ii) whether each item is relevant to the title.
\end{enumerate}

Formally, let $\mathrm{EVAL}$ be a function that critiques a candidate title $M_{\mathrm{title}}^{(0)}$ with respect to the item set $\mathcal{I}$ and the constraints in Eqs.\,\ref{eq:obj:max-items}--\ref{eq:constr:length}. The feedback is then used to refine the generation:
\begin{equation}
\label{eq:refine}
\begin{aligned}
&\text{Feedback} 
\;=\; 
\mathrm{EVAL}\bigl(M_{\mathrm{title}}^{(0)},\, I_1,\dots,I_N\bigr),\\
&M_{\mathrm{title}}^{(1)}
\;=\;
\mathrm{GEN}\Bigl(\text{Feedback},\, M_{\mathrm{title}}^{(0)},\, \{K_i^j\}\Bigr).
\end{aligned}
\end{equation}

We refer to this iterative process as \emph{Refinement Circle}, where the generator agent refines the title based on feedback to improve coverage while maintaining the length constraint $K$.  While prior work mostly stops after a predefined number of iterations \citep{wu2023autogen, yao2022react}, our key contribution is a theoretical bound on the number of iterations $T$ required to approximate the optimal title. Specifically, we analyze the convergence behavior of the refinement circle and formally characterize how quickly it approaches the optimal title defined in Eq.,\ref{eq:obj:max-items}. We present the detailed analysis in Section~\ref{sec:bound for optimal title}.

\subsection{Arbitration Agents}
\label{sec:mod_arb}
Since LLM sampling can yield diverse outputs, the system generates $k$ candidate titles by executing the \emph{GAG} and \emph{Refinement Circle} stages multiple times:
\[
    \bigl\{M_{\mathrm{title}}^{(1)}, M_{\mathrm{title}}^{(2)}, \dots, M_{\mathrm{title}}^{(k)}\bigr\},
\]
for the same set of items $\mathcal{I}$. Each candidate title is accompanied by a corresponding ``reasoning trace'' (e.g., chain-of-thought or justification), and may differ in both coverage (i.e., relevance to the items) and length.

An arbitrator agent selects the best module title among the $k$ candidates. Let $\mathrm{ARBIT}$ be an LLM-driven function that takes as input the $k$ candidate titles and the moderator’s summary $S_{\mathrm{mod}}$, and outputs a single best title:
\begin{equation}
\label{eq:arb}
    M_{\mathrm{final}}
    \;=\;
    \mathrm{ARBIT}\Bigl(
        \{M_{\mathrm{title}}^{(1)}, \dots, M_{\mathrm{title}}^{(k)}\}, 
        \; S_{\mathrm{mod}}
    \Bigr).
\end{equation}
To approximate the objective in Eqs.\,\ref{eq:obj:max-items}--\ref{eq:constr:length}, the arbitrator may internally rank each title based on (i) Coverage: The number of items $I_i$ for which the title is deemed relevant, (ii) Length Compliance: Whether (or how closely) the title adheres to the character limit $K$, (iii) Stylistic or Additional Constraints: Including any platform or language guidelines.

Thus, $M_{\mathrm{final}}$ is the final module title displayed in the carousel, selected to balance item coverage and constraint satisfaction in accordance with our optimization framework. An algorithmic summary of the CARTS framework is provided in Appendix~A.

\begin{table*}[h]
    \centering
    \caption{Compare CARTS with benchmarks.}
    \small
    \begin{tabular}{lllllllll}
        \toprule
        \multicolumn{1}{c}{\multirow{2}{*}{\textbf{Benchmarks}}} & \multicolumn{2}{c}{\textbf{Beauty}} & \multicolumn{2}{c}{\textbf{Electronics}} & \multicolumn{2}{c}{\textbf{Fashion}} & \multicolumn{2}{c}{\textbf{Home and Kitchen}} \\
        \cmidrule(lr){2-3} \cmidrule(lr){4-5} \cmidrule(lr){6-7} \cmidrule(lr){8-9}
        \multicolumn{1}{c}{}                           & \textbf{LLM judge}     & \textbf{BERT}    & \textbf{LLM judge}        & \textbf{BERT}       & \textbf{LLM judge}   & \textbf{BERT}  & \textbf{LLM judge}   & \textbf{BERT}   \\
        \midrule
        Vanilla                          &    0.73           &    0.56             & 0.852                 &        0.572           &     0.739        &    0.591          &   0.831          &    0.562           \\
        DRE                                            & 0.745         &   0.582              & 0.851            &   0.602                & 0.779       &  0.6            & 0.853       &   0.57            \\
        CoT                                        & 0.744         &   0.57              & 0.866            &  0.578                 & 0.751       &     0.6         & 0.845       &    0.567           \\							
        LLM-CoT                          &    0.809           &    0.57             & 0.865                 &    0.575               &   0.835          &  0.589            &   0.872          &    0.57           \\
        CARTS                                          &    0.891           &  0.634               &    0.939              &    0.655               &    0.928         &      0.70        &    0.953         &   0.631  \\
        \bottomrule
    \end{tabular}
    \label{tab:offline}
\end{table*}

\subsection{Approximation Guarantee of CARTS}\label{sec:bound for optimal title}
In this subsection, we derive a theoretical bound on the number of iterations $T$ required to approximate the optimal title during \emph{Refinement Circle}. Let $\mathcal{C}$ denote the set of all feasible titles satisfying the length constraint $\mathrm{len}(M) \le K$. We define the optimal achievable coverage as
\[
\mathrm{OPT} := \max_{M \in \mathcal{C}} \sum_{i \in \mathcal{N}} R(M, I_i),
\]
where $R(M, I_i) \in \{0, 1\}$ indicates whether the title $M$ covers item $I_i$.

To enable the convergence analysis, we introduce two assumptions:
\begin{assumption}[$\beta$‐reliable feedback agent]
\label{ass:evaluator+}
At iteration $m$ let
$
  C_m=\sum_{i\in\mathcal N} R\!\bigl(M_{\mathrm{title}}^{(m)},I_i\bigr)
$
be the current coverage.
If $C_m<\mathrm{OPT}$, the feedback agent outputs a non–empty set
$
  \mathcal U_m\subseteq\mathcal N
$
such that
$
  \exists I_j\in\mathcal U_m
$
with
$
  R(M_{\mathrm{title}}^{(m)},I_j)=0
$
(i.e.\ \emph{at least one} uncovered item is flagged)
with probability at least $\beta>0$.
\end{assumption}

\begin{assumption}[$\gamma$‐reliable generator agent]
\label{ass:generator+}
Conditioned on $\mathcal U_m$,
the generator produces a refined title
$M_{\mathrm{title}}^{(m+1)}$ satisfying
\begin{enumerate}
  \item $C_{m+1}\;-\;C_m\;\ge\;1$
        \hfill\textbf{(covers \emph{at least one} new item);}  
  \item $\mathrm{len}\!\bigl(M_{\mathrm{title}}^{(m+1)}\bigr)\le K$;
  \item $R(M_{\mathrm{title}}^{(m+1)},I_i)\ge
         R(M_{\mathrm{title}}^{(m)},I_i)$ for all $i$
        \hfill\textbf{(no regression).}
\end{enumerate}
The generator succeeds with probability $\gamma>0$.
\end{assumption}

\begin{theorem}[Approximate optimal coverage in $T$ rounds]
\label{thm:hp+}
Fix $\alpha\in(0,1]$ and $\varepsilon\in(0,1)$.
Let $p:=\beta\gamma$ and define
\[
  \Lambda(\alpha,\beta,\gamma,\mathrm{OPT},\varepsilon)
  \;:=\;
  \Bigl\lceil
     \frac{\alpha\,\mathrm{OPT}-C_0}{p}
     \;+\;
     \frac{2\ln(1/\varepsilon)}{p}
  \Bigr\rceil .
\]
Under Assumptions~\ref{ass:evaluator+}–\ref{ass:generator+},
running the generator–feedback loop for
$
  T\;\ge\;\Lambda(\alpha,\beta,\gamma,\mathrm{OPT},\varepsilon)
$
iterations guarantees
\[
  \Pr\!\bigl[C_T\;\ge\;\alpha\,\mathrm{OPT}\bigr]\;\ge\;1-\varepsilon .
\]
\end{theorem}

\begin{proof}
For $m\ge0$ define
\[
  Y_{m+1}\;:=\;\mathbf 1\!\bigl[C_{m+1}>C_m\bigr].
\]
By Assumptions~\ref{ass:evaluator+}–\ref{ass:generator+},
whenever $C_m<\mathrm{OPT}$
\[
  \Pr\!\bigl[Y_{m+1}=1\,\mid\,\text{history}\bigr]
  \;\ge\;p=\beta\gamma ,
\]
and the $Y_{m+1}$ are conditionally independent across rounds.

Let
$
  S_T:=\sum_{m=1}^{T} Y_m
$
be the number of rounds that achieve \emph{any} improvement.
Since each such round covers at least one additional item and never
loses coverage,
\[
  C_T
  \;\ge\;
  C_0+S_T .
\]
Thus
$
  C_T\ge\alpha\,\mathrm{OPT}
$
is implied by
$
  S_T\ge\alpha\,\mathrm{OPT}-C_0 .
$

Each $Y_{m+1}$ stochastically dominates
$\operatorname{Bernoulli}(p)$, so
$S_T$ dominates $\operatorname{Binomial}(T,p)$.
Consequently it suffices to bound a Binomial lower tail.

For $Z\sim\operatorname{Binomial}(T,p)$ with mean $\mu=pT$ and any
$\delta\in(0,1)$, by Chernoff bound,
\[
  \Pr[Z\le(1-\delta)\mu]
  \;\le\;
  \exp\!\bigl(-\tfrac{\delta^{2}\mu}{2}\bigr).
\]
Choose $\delta$ so that
$
  (1-\delta)\mu
  =
  \alpha\,\mathrm{OPT}-C_0,
$
i.e.
$
  \delta
  =
  1-\tfrac{\alpha\,\mathrm{OPT}-C_0}{pT}.
$
Requiring the bound to be $\le\varepsilon$ gives
\[
  pT-(\alpha\,\mathrm{OPT}-C_0)
  \;\ge\;
  2\ln(1/\varepsilon),
\]
which is equivalent to
$
  T\ge\Lambda(\alpha,\beta,\gamma,\mathrm{OPT},\varepsilon).
$
At this $T$
\(
  \Pr[S_T<\alpha\,\mathrm{OPT}-C_0]\le\varepsilon,
\)
so
$
  \Pr[C_T\ge\alpha\,\mathrm{OPT}]\ge1-\varepsilon.
$
\end{proof}

\begin{corollary}[Expected iterations to achieve $\mathrm{OPT}$]
\label{cor:exp+}
Let $T_{\mathrm{OPT}}:=\min\{m:C_m=\mathrm{OPT}\}$ and set
$
  T^\star=(\mathrm{OPT}-C_0)/(\beta\gamma).
$
Then under
Assumptions~\ref{ass:evaluator+}–\ref{ass:generator+}
\[
  \mathbb E[T_{\mathrm{OPT}}]
  \;\le\;
  \frac{\mathrm{OPT}-C_0}{\beta\gamma}
  \;+\;
  \frac{2}{\beta\gamma}.
\]
\end{corollary}

\begin{proof}
Proof details can be seen in Appendix D.
\end{proof}

\section{Experiments}\label{sec:exp}
\subsection{Datasets}
For our experimental evaluation, we utilize the Amazon Review dataset \cite{hou2024bridging}  across four categories: Beauty, Electronics, Fashion, and Home and Kitchen. We randomly sampled 1,000 anchor items from each category. For each anchor item, we generated 10 similar recommended items, yielding a total of 40,000 recommended items across all categories. This comprehensive sampling strategy ensured a diverse and representative dataset for our study. Similar item recommendations were generated using an approximate nearest neighbor (ANN) model in the item embedding space. Item embedding vectors are derived from item categories, titles, and descriptions with the Universal Sentence Encoder (USE) model \cite{cer2018universal}. 


\subsection{Benchmark Models and Evaluation Metrics}

To demonstrate the efficacy of our proposed framework, we compare its performance against four benchmark models: Vanilla GPT (Vanilla), which utilizes a single vanilla LLM call to generate the module title in a single step; Chain of Thought (CoT), which enhances the generation process by instructing the model to think step-by-step, following the method proposed by \citet{kojima2022large}; LLM-guided Chain of Thought (LLM-CoT), which first prompts the LLM to explicitly outline the reasoning steps required for the title generation task before executing the task, as described in \citet{zhang2022automatic}; and the Data-level Recommendation Explanation (DRE) framework, which generates keywords for individual items first and then creates the title based on these keywords \citep{gao2024dre}.

We evaluate the quality and relevance of the generated titles using two metrics.



\begin{itemize}
\item \textbf{BERT Score:} Computes semantic similarity between the generated title and each item's description using token-level cosine similarity of BERT embeddings. We calculate the score for each item-title pair and report the module-level average. Higher scores indicate stronger relevance.

\item \textbf{LLM Judge Score:} Uses GPT-4o \citep{chatgpt}, guided by the Chain of Steps prompting strategy \citep{zheng2024judging}, to assess if a title accurately represents each item. The model outputs 1 (relevant) or 0 (not), and the module-level score is the average across items.
\end{itemize}


\begin{table*}[h]
    \centering
    \caption{Ablation studies of CARTS over four categories.}
    \small
    \begin{tabular}{lllllllll}
        \toprule
        \multicolumn{1}{c}{\multirow{2}{*}{\textbf{Benchmarks}}} & \multicolumn{2}{c}{\textbf{Beauty}} & \multicolumn{2}{c}{\textbf{Electronics}} & \multicolumn{2}{c}{\textbf{Fashion}} & \multicolumn{2}{c}{\textbf{Home and Kitchen}} \\
        \cmidrule(lr){2-3} \cmidrule(lr){4-5} \cmidrule(lr){6-7} \cmidrule(lr){8-9}
        \multicolumn{1}{c}{}                           & \textbf{LLM judge}     & \textbf{BERT}    & \textbf{LLM judge}        & \textbf{BERT}       & \textbf{LLM judge}   & \textbf{BERT}  & \textbf{LLM judge}   & \textbf{BERT}   \\
        \midrule
        CARTS w.t. refinement                       & 0.827          &   0.633              &    0.9215        &     0.652              &  0.876      &  0.702            &   0.927     &   0.627            \\
        CARTS w.t. arbitrator                          &   0.845            &   0.633              &   0.93               &   0.653                &     0.889        &       0.701       &      0.9345       &     0.628          \\
        CARTS                                          &    0.891           &  0.634               &    0.939              &    0.655               &    0.928         &      0.70        &    0.953         &   0.631  \\
        \bottomrule
    \end{tabular}
    \label{tab:ablation}
\end{table*}

\subsection{Experiment Results with Benchmarks}
Table~\ref{tab:offline} presents the performance comparison of four methods—Vanilla, CoT, LLM-CoT, and our proposed method CARTS—on the Amazon Review dataset across four product categories.
We report BERT scores and LLM judge scores using GPT-4o for the generated title results.

The Vanilla method, a single LLM call without explicit guidance, exhibits the lowest performance across all categories on both metrics. It often produces generic titles that lack comprehensive item coverage, leading to lower semantic alignment (BERT Score) and reduced LLM-judged relevance.

The DRE method, which distills keywords from item information to guide title generation, consistently improves upon Vanilla. It shows an LLM judge score improvement of 2-5\% and a BERT score improvement of 1.4-5.3\%, demonstrating the benefit of providing structured input for the LLM (except Electronics LLM judge score).

The CoT approach enhances performance by prompting the model for step-by-step reasoning, improving LLM judge scores by 1-2\% and BERT scores by 1-1.6\% over Vanilla. This guided reasoning helps identify more salient features. However, CoT's single-agent, single-pass nature still limits its ability to fully capture diverse attributes or correct misalignments.

LLM-CoT explicitly separates reasoning and generation, leading to significant gains in LLM judge scores compared to CoT (e.g., Beauty: 0.744 to 0.809; Fashion: 0.751 to 0.835). However, its BERT scores do not show a corresponding improvement, and its effectiveness remains bounded by internal planning without external feedback for iterative refinement.

Our proposed framework, CARTS, achieves the highest performance across all categories and metrics. It significantly improves both LLM judge (12-26\%) and BERT (10-18\%) scores compared to Vanilla. 

In summary, the offline results demonstrate a consistent and substantial performance gap, validating CARTS's design choices and highlighting the importance of agent collaboration and feedback-driven optimization for effective recommendation explanation.

\begin{table*}[h]
    \centering
    \caption{Comparison of different LLMs on LLM Judge Score and BERT Score across four categories}
    \small
    \begin{tabular}{l cc cc cc cc}
        \toprule
        \multirow{2}{*}{\textbf{LLM}} & \multicolumn{2}{c}{\textbf{Beauty}} & \multicolumn{2}{c}{\textbf{Electronics}} & \multicolumn{2}{c}{\textbf{Fashion}} & \multicolumn{2}{c}{\textbf{Home}} \\
        \cmidrule(lr){2-3} \cmidrule(lr){4-5} \cmidrule(lr){6-7} \cmidrule(lr){8-9}
        & \textbf{LLM Judge} & \textbf{BERT} & \textbf{LLM Judge} & \textbf{BERT} & \textbf{LLM Judge} & \textbf{BERT} & \textbf{LLM Judge} & \textbf{BERT} \\
        \midrule
        GPT-4o           & 0.891 & 0.634 & 0.939 & 0.655 & 0.928 & 0.700 & 0.953 & 0.631 \\
        Gemini-2.0-Flash & 0.850 & 0.581 & 0.870 & 0.566 & 0.840 & 0.640 & 0.860 & 0.617 \\
        Gemini-1.5-Flash & 0.820 & 0.630 & 0.840 & 0.566 & 0.810 & 0.656 & 0.830 & 0.567 \\
        LLaMA-3          & 0.790 & 0.611 & 0.810 & 0.556 & 0.780 & 0.553 & 0.800 & 0.565 \\
        GPT-3.5          & 0.600 & 0.600 & 0.620 & 0.641 & 0.580 & 0.695 & 0.610 & 0.565 \\
        \bottomrule
    \end{tabular}
    \label{tab:llm-bert-judge-comparison}
\end{table*}

\subsection{Ablation Study}
We conducted an ablation study to systematically evaluate the individual contributions of the critical components within CARTS. 
Specifically, we compared three variants: (i) CARTS without refinement cycle, testing the impact of the iterative evaluation-refinement loop; (ii) CARTS without arbitrator, examining the effect of removing the final decision-making agent; and (iii) the complete CARTS framework.

Results demonstrate that removing the refinement cycle significantly reduces LLM Judge scores across all categories. LLM judge scores reduces between 1.9\% to 7.2\% for four categories, underscoring the importance of iterative refinement in enhancing the relevance. 
Similarly, removing the arbitrator agent caused a noticeable decline in module title relevance performance, LLM judge score reduces between 1\% to 5.4\% for four categories, highlighting that explicitly selecting the best candidate title among multiple refined options is crucial for achieving optimal coverage. The complete CARTS framework consistently delivered the highest scores, confirming the indispensable role of each component. 




\subsection{Comparision with Different LLM on Summarazation}
Table~\ref{tab:llm-bert-judge-comparison} compares five LLMs using BERT Score and LLM Judge Score across four categories. GPT-4o achieves the highest performance on both metrics, with BERT Scores ranging from 0.631 to 0.700 and LLM Judge Scores consistently above 0.89. This indicates strong semantic coverage and item-level relevance, making GPT-4o the most reliable model for our task.

Gemini-2.0-Flash ranks second, outperforming Gemini-1.5-Flash across all categories. LLaMA-3 follows with moderate performance. In contrast, GPT-3.5 shows the weakest results, with LLM Judge Scores around 0.60 and BERT Scores between 0.565 and 0.695, reflecting poor semantic alignment and relevance.

Moreover, we observe that all models except GPT-3.5 consistently respect the character-length constraint, which is critical in real-world UI applications. Based on these results, we adopt GPT-4o as the final model in our framework due to its superior alignment, relevance, and constraint adherence.

\subsection{Case studies}
In Figure ~\ref{fig:examples}, we show two examples of module titles generated with CARTS. In Figure \ref{fig:examples}(A), this carousel shows a list of Apple MacBook laptops, including Pro and Air two models. The generated title ``\textbf{Sleek and high-performance MacBook Pro and Air}" successfully captures the two models and highlights the two key advantages of them (``sleek" and ``high performance"). In Figure \ref{fig:examples}(B), this carousel presents a list of portable speakers from various brands. Instead of a simple summarization ``Portable Speakers", CARTS provides a carousel title showing the use scenarios of these speakers: "Outdoor and Party Music", making it more engaging and informative.

\begin{figure*}[h]
    \centering    
    \includegraphics[width=1.0\linewidth]{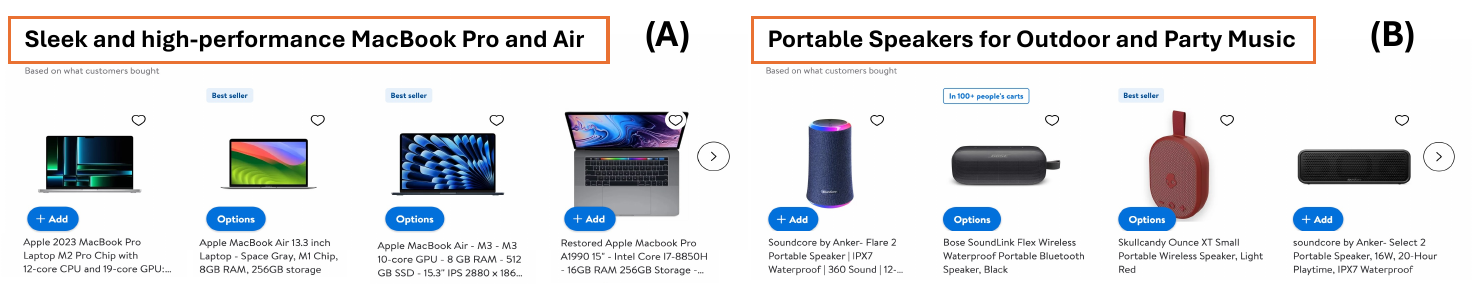}
    \caption{Two module title examples.}
    \label{fig:examples}
\end{figure*}





\section{Online A/B Test Results}
To evaluate the impact of generated module titles for customer engagement and business metrics, we also conducted an A/B test experiment for generated in-house module title results and report the business related metrics, including lifts on click-through rate (CTR), add-to-cart rate (ATCR), and gross merchandise value (GMV). Figure~\ref{fig:examples} shows two module title examples launched in the experiment.

In the A/B testing, the control is a black-box recommendation explanation model, while the variant has the module specific titles generated by CARTS. The traffic was split as 50-50 between control and the variant. The A/B testing results are shown in Table~\ref{tab:ab}. 

The results demonstrate statistically significant increase in CTR, ATCR, and GMV metrics. Specifically, the CTR has an uplift of 0.8\% and ATCR has an uplift of 6.28\%, indicating higher customer engagement with our generated module titles. In addition, the GMV, reflecting the total sales value, showed a remarkable increase of 3.78\%. These results confirm that the generated module titles not only enhance customer-product interactions but also drive more sales for the e-commerce platform. Please note that due to proprietary data protection agreement, details about actual dollar value implementation costs and business metrics cannot be published in the paper, however, increased business metrics far exceed the cost of running our LLM-baed multi-agent pipeline to generate the eCommerce module titles.

\begin{table}[h]
    \centering
    \caption{A/B testing evaluation results. The lift results are percentages with 95\% confidence interval.}
    \label{tab:ab}
    \begin{tabular}{cccc}
        \hline
         & CTR & ATCR & GMV \\
        \hline
        Lift & 0.8 $\pm$ 0.46 &  6.28 $\pm$ 2.07  & 3.78 $\pm$ 0.19 \\
        \hline
    \end{tabular}
\end{table}

\section{Conclusion}
We presented CARTS, a multi-agent LLM framework for generating concise and relevant module-level titles in recommendation systems. Unlike existing multi-agent approaches developed for unstructured domains like law or finance, CARTS addresses the unique challenges of structured input, length constraints, and engagement-driven objectives in recommender settings. By decomposing the task into generation, feedback and arbitration agents, and introducing a Generation-Augmented Generation (GAG) strategy, CARTS achieves high semantic coverage under practical constraints. We further provide theoretical guarantees on convergence to near-optimal coverage. Extensive offline experiments and real-world A/B testing confirm CARTS’s effectiveness in improving title relevance, CTR, and GMV, demonstrating the practical value of collaborative LLM agents in recommendation workflows.

\clearpage

\bibliography{title}
\bibliographystyle{icml2025}

\newpage
\appendix
\onecolumn
\section*{Appendix A: Algorithmic Summary}
Algorithm\,\ref{alg:CARTS} summarizes the entire \textbf{CARTS} procedure in pseudocode. This procedure approximates the objective of maximizing relevance coverage (i.e., number of items matched) while respecting a character limit $K$ and any additional constraints (see Eqs.\,\ref{eq:obj:max-items}--\ref{eq:constr:others}).

Note that under CARTS by prompting an LLM to extract only $l$ keywords, CARTS highlights each item's essential attributes, increasing the chance of capturing shared properties across items. Multiple language agents collaborate iteratively. A \texttt{feedback} agent critiques each intermediate title for coverage (number of relevant items) and length adherence ($\le K$), while the \texttt{generator} agent refines accordingly. In addition, generating $k$ distinct titles enables diversity. The \texttt{moderator} summarizes each candidate’s coverage, length compliance, and reasoning. Finally, the \texttt{arbitrator} selects the single best title that best satisfies the optimization criteria.

In Section\,\ref{sec:exp}, we show that CARTS significantly improves the relevance and transparency of generated module titles for carousel recommendations, enhancing both offline coverage metrics and online user engagement.

\begin{algorithm*}[h]
\caption{CARTS: Multi-Agent Generation Augmented Generation}
\label{alg:CARTS}
\begin{algorithmic}[1]
\STATE {\bfseries Input:} A set of items $\mathcal{I} = \{I_1, \ldots, I_N\}$, number of keywords $l$, number of title samples $k$, length limit $K$
\STATE {\bfseries Output:} Final module title $M_{\mathrm{final}}$

\vspace{1ex}
\STATE \textbf{1. Keyword Distillation}
\STATE \textbf{for} $i = 1$ to $N$ \textbf{do}
    \STATE \quad $\{K_i^1, \dots, K_i^l\} \gets \mathrm{DISTILL}(C_i, T_i, P_i)$
\STATE \textbf{end for}

\vspace{1ex}
\STATE \textbf{2. Initial Title Generation}
\STATE $M_{\mathrm{title}}^{(0)} \gets G_{\mathrm{title}}\bigl(\{(I_i, K_i^{1:l})\}_{i=1}^N\bigr)$
\COMMENT{Prompt includes coverage and length reminders.}

\vspace{1ex}
\STATE \textbf{3. Refinement Circle}
\STATE \textbf{for} $m = 1$ to $k$ \textbf{do}
    \STATE \quad $\mathrm{Feedback}^{(m)} \gets \mathrm{EVAL}\Bigl(M_{\mathrm{title}}^{(m-1)},\, \mathcal{I},\, K\Bigr)$
    \COMMENT{Checks coverage \& length.}
    \STATE \quad $M_{\mathrm{title}}^{(m)} \gets \mathrm{GEN}\Bigl(\mathrm{Feedback}^{(m)},\, M_{\mathrm{title}}^{(m-1)},\, \{K_i^j\}\Bigr)$
\STATE \textbf{end for}

\vspace{1ex}
\STATE \textbf{4. Arbitration}
\STATE $S_{\mathrm{mod}} \gets \mathrm{MOD}\Bigl(\{(M_{\mathrm{title}}^{(m)}, R^{(m)})\}_{m=1}^{k}\Bigr)$
\STATE $M_{\mathrm{final}} \gets \mathrm{ARBIT}\Bigl(\{M_{\mathrm{title}}^{(m)}\}_{m=1}^{k},\, S_{\mathrm{mod}},\, K\Bigr)$
\COMMENT{Selects title maximizing coverage while respecting $K$.}

\vspace{1ex}
\STATE \textbf{return} $M_{\mathrm{final}}$
\end{algorithmic}
\end{algorithm*}

\section*{Appendix B: Sample Generation Results for CARTS}
\begin{figure}
    \centering
    \includegraphics[width=\linewidth]{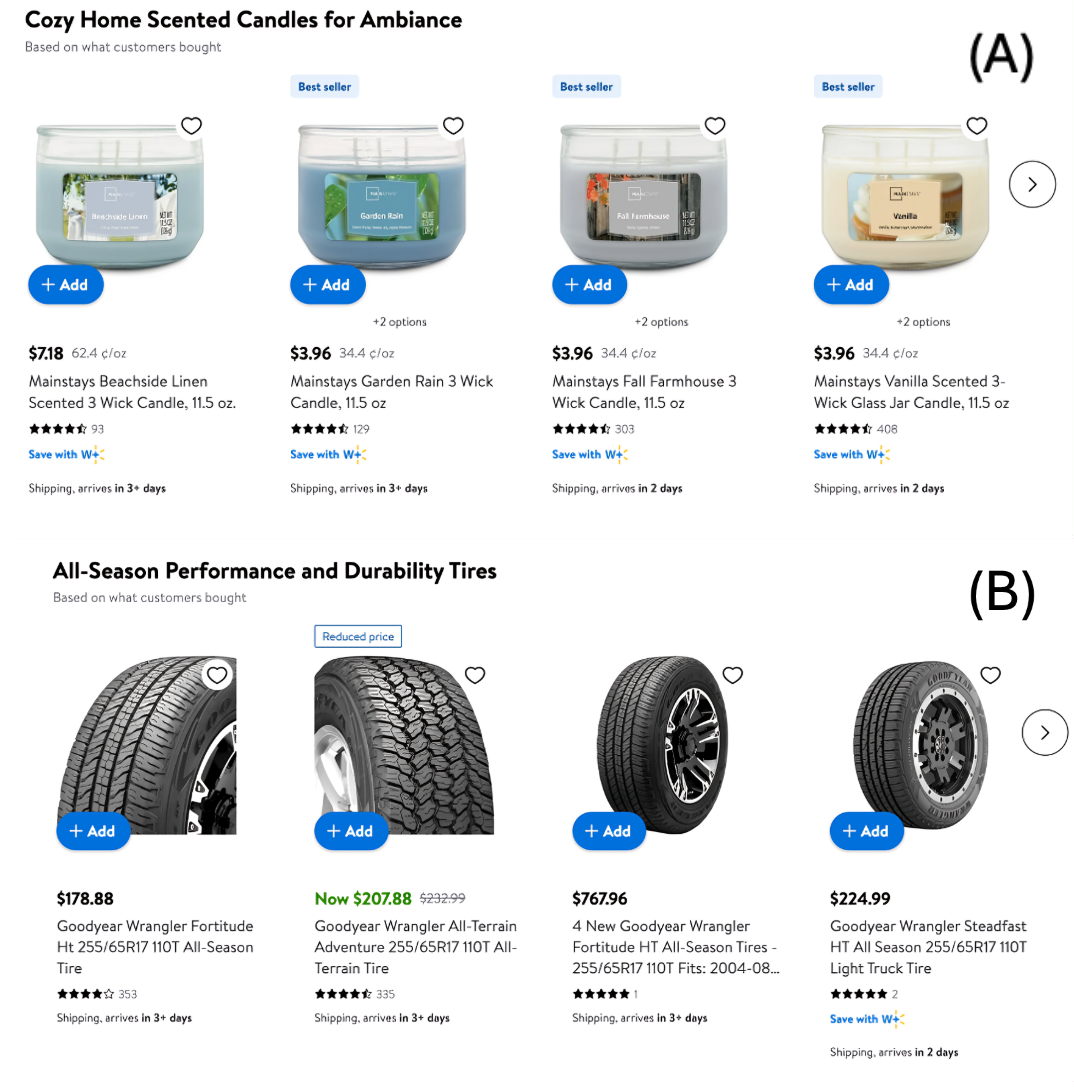}
    \caption{Sample Generation Results for CARTS. Examples from (A) Candles, (B) Tires}
    \label{fig:enter-labe-1}
\end{figure}

\begin{figure}
    \centering
    \includegraphics[width=\textwidth]{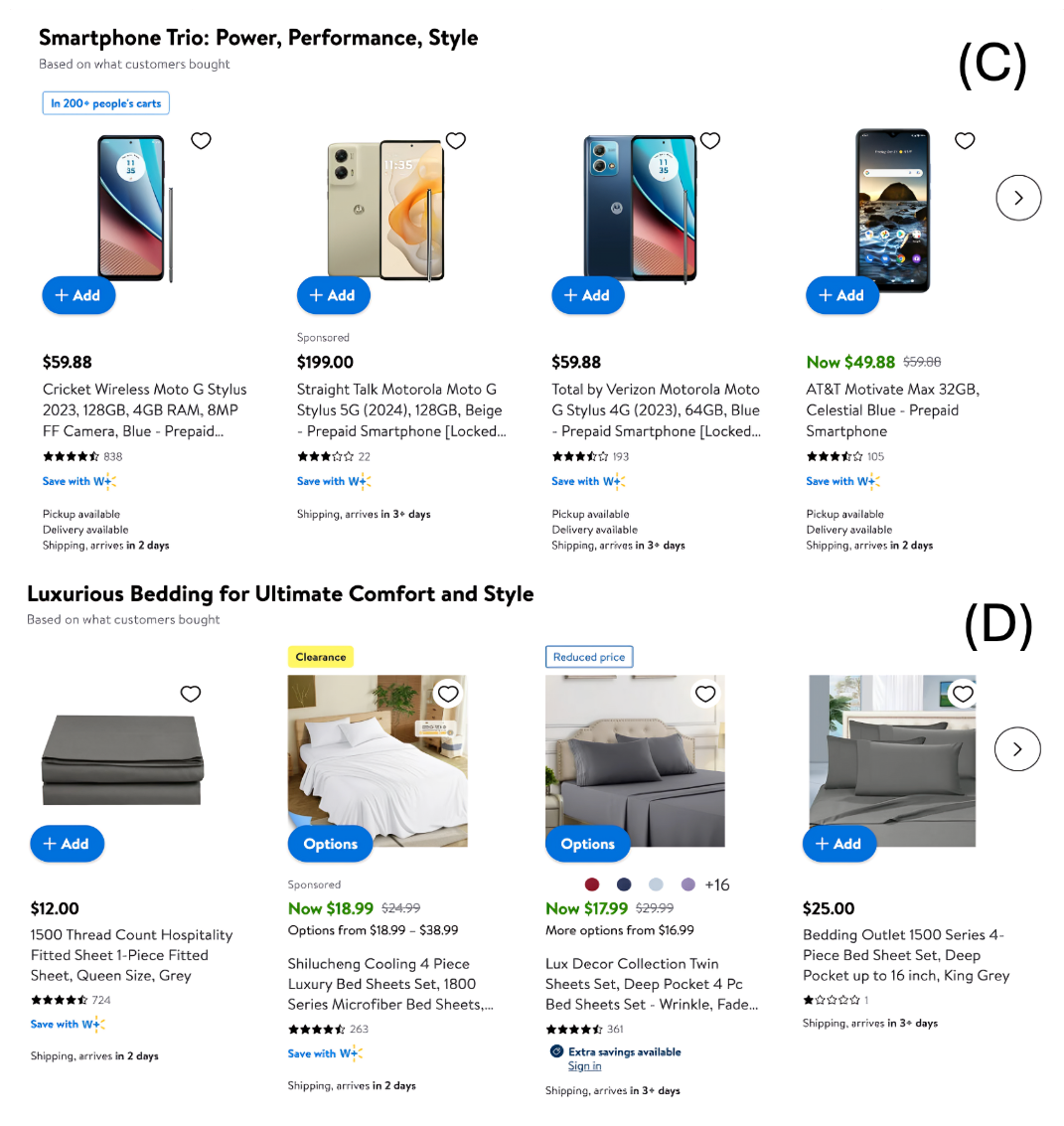}
    \caption{Sample Generation Results for CARTS. Examples from (C) Smart Phones, and (D) Bedding Products}
    \label{fig:enter-label-2}
\end{figure}

\newpage
\section*{Appendix C: prompts for CARTS}
\subsection*{Keywords generation prompt}
\begin{tcolorbox}
You are an English speaking eCommerce catalog specialist. You are an expert in generating keywords for a given product. \\

Consider the following product: \\
\{prod\_info\} \\

Output at most 5 short keywords relevant to the product. \\

Please just output the keywords and separate them with commas. \\
Do not add any other text.\\
\end{tcolorbox}

\subsection*{GAG prompt}
\begin{tcolorbox}

You are an eCommerce specialist. Your expertise is in generating a title for a group of items presented in an eCommerce module. \\
Your task is to name this eCommerce module.\\

The list of products and their associated keywords are: \\
\{prod\_info\_and\_keys\}\\

Generate a module title for the list of items to explain them, aiming to increase customer engagement and improve eCommerce module transparency. \\

Restrict the response to the following format strictly:\\

"title: A maximum of 10 word title that is relevant to the list of items."\\
\end{tcolorbox}

\subsection*{Feedback Prompt}
\begin{tcolorbox}
You are an evaluator that evaluates the generated titles for a group of items. \\

Here is the generated title: \\
\{title\} \\
Here is the set of items and their associated keywords that the title is generated for them:\\\{prod\_info\_and\_keys\}\\

Determine if the title is relevant enough to some or all of the items and can increase customer engagement or improve ecommerece module transparency. \\
If so, provide concise feedback pointing to at least one such uncovered item that the generator could improve on.

Keep the feedback within 30 words. 
\end{tcolorbox}

\subsection*{Regeneration Prompt}
\begin{tcolorbox}
You are an eCommerce title generation specialist working in a refinement circle. Your task is to improve a previously generated title based on feedback from an evaluator. \\

The list of items and their associated keywords are: \\
\{prod\_info\_and\_keys\} \\

The original title was: \\
\{title\} \\

The evaluator provided the following feedback: \\
\{feedback\} \\

Generate a refined title that \\
(1) addresses at least one uncovered or weakly covered item indicated in the feedback\\
(2) preserves all existing item coverage in the original title. \\
(3) Ensure the revised title is no more than 10 words and formatted exactly as follows: \\
title: <your refined title>





\end{tcolorbox}

\subsection*{Arbitrator Prompt}
\begin{tcolorbox}
You are an eCommerece specialist tasked with selecting the best title for an eCommerce module after refinement. \\
You are given two titles: \\

1. The original title: {title} \\
2. The refined title: {title\_2} \\

These titles are generated based on the following list of products and their associated keywords: \\
\{prod\_info\_and\_keys\}\\

Select the title that is more relevant and likely to increase customer engagement and improve module transparency. \\
Output only the selected title.
\end{tcolorbox}

\section*{Appendix D: proof of collary 3.4}
\begin{proof}
Apply Theorem~\ref{thm:hp+} with $\alpha=1$.
For any integer $k\ge0$ set
\(
\varepsilon_k:=e^{-k}.
\)
Then with
\(
T_k:=\Lambda(1,\beta,\gamma,\mathrm{OPT},\varepsilon_k)
=T^\star+\tfrac{2k}{\beta\gamma}
\)
we have
\(
\Pr[T_{\mathrm{OPT}}\;>\;T_k]\;\le\;e^{-k}.
\)
Using the tail–sum representation
$
\mathbb{E}[T_{\mathrm{OPT}}]
  = \sum_{t\ge0} \Pr[T_{\mathrm{OPT}}>t]
$
and splitting the sum into blocks of length
$2/(\beta\gamma)$
gives
\[
\mathbb{E}[T_{\mathrm{OPT}}]
  \;\le\;
  T^\star
  \;+\;
  \frac{2}{\beta\gamma}\sum_{k\ge0} e^{-k}
  \;\le\;
  T^\star+\frac{2}{\beta\gamma}.
\]
\end{proof}

\end{document}